\documentclass[pra,aps,showpacs,twocolumn,twoside,superscriptaddress]{revtex4}

\usepackage{amsmath,amsfonts,amssymb,mathrsfs,hyperref,color,epsfig,graphics,graphicx,latexsym,mathrsfs,revsymb,theorem,url,verbatim,epstopdf,subfigure}

\usepackage{hyperref}

\hypersetup{colorlinks,linkcolor={blue},citecolor={blue},urlcolor={red}}

\newtheorem{proposition}{Proposition}
\newtheorem{lemma}{Lemma}

\newtheorem{theorem}{Theorem}
\newtheorem{corollary}{Corollary}

\newtheorem{example}{Example}

\def\squareforqed{\hbox{\rlap{$\sqcap$}$\sqcup$}}
\def\qed{\ifmmode\squareforqed\else{\unskip\nobreak\hfil
\penalty50\hskip1em\null\nobreak\hfil\squareforqed
\parfillskip=0pt\finalhyphendemerits=0\endgraf}\fi}
\def\endenv{\ifmmode\;\else{\unskip\nobreak\hfil
\penalty50\hskip1em\null\nobreak\hfil\;
\parfillskip=0pt\finalhyphendemerits=0\endgraf}\fi}
\newenvironment{proof}{\noindent \textbf{{Proof.~} }}{\qed}
\def\Dbar{\leavevmode\lower.6ex\hbox to 0pt
{\hskip-.23ex\accent"16\hss}D}
\def\bpf{\begin{proof}}
\def\epf{\end{proof}}

\newcommand{\bra}[1]{\langle{#1}|}
\newcommand{\ket}[1]{|{#1}\rangle}
\newcommand{\proj}[1]{|{#1}\rangle \langle {#1}|}

\newcommand{\nc}{\newcommand}

\def\bea{\begin{eqnarray}}
\def\eea{\end{eqnarray}}
\def\beq{\begin{equation}}
\def\eeq{\end{equation}}

\def\min{\mathop{\rm min}}

\def\ra{\rightarrow}

\def\ox{\otimes}

\def\a{\alpha}

\def\r{\rho}
\def\s{\sigma}

\def\c{\chi}

\nc{\bbA}{\mathbb{A}} \nc{\bbB}{\mathbb{B}} \nc{\bbC}{\mathbb{C}}
\nc{\bbD}{\mathbb{D}} \nc{\bbE}{\mathbb{E}} \nc{\bbF}{\mathbb{F}}
\nc{\bbG}{\mathbb{G}} \nc{\bbH}{\mathbb{H}} \nc{\bbI}{\mathbb{I}}
\nc{\bbJ}{\mathbb{J}} \nc{\bbK}{\mathbb{K}} \nc{\bbL}{\mathbb{L}}
\nc{\bbM}{\mathbb{M}} \nc{\bbN}{\mathbb{N}} \nc{\bbO}{\mathbb{O}}
\nc{\bbP}{\mathbb{P}} \nc{\bbQ}{\mathbb{Q}} \nc{\bbR}{\mathbb{R}}
\nc{\bbS}{\mathbb{S}} \nc{\bbT}{\mathbb{T}} \nc{\bbU}{\mathbb{U}}
\nc{\bbV}{\mathbb{V}} \nc{\bbW}{\mathbb{W}} \nc{\bbX}{\mathbb{X}}
\nc{\bbZ}{\mathbb{Z}}

\nc{\bp}{{\bf p}}
\nc{\bq}{{\bf q}}
\nc{\br}{{\bf r}}

\nc{\bA}{{\bf A}} \nc{\bB}{{\bf B}} \nc{\bC}{{\bf C}}
\nc{\bD}{{\bf D}} \nc{\bE}{{\bf E}} \nc{\bF}{{\bf F}}
\nc{\bG}{{\bf G}} \nc{\bH}{{\bf H}} \nc{\bI}{{\bf I}}
\nc{\bJ}{{\bf J}} \nc{\bK}{{\bf K}} \nc{\bL}{{\bf L}}
\nc{\bM}{{\bf M}} \nc{\bN}{{\bf N}} \nc{\bO}{{\bf O}}
\nc{\bP}{{\bf P}} \nc{\bQ}{{\bf Q}} \nc{\bR}{{\bf R}}
\nc{\bS}{{\bf S}} \nc{\bT}{{\bf T}} \nc{\bU}{{\bf U}}
\nc{\bV}{{\bf V}} \nc{\bW}{{\bf W}} \nc{\bX}{{\bf X}}
\nc{\bZ}{{\bf Z}}

\nc{\bmA}{{\bm A}} \nc{\bmB}{{\bm B}} \nc{\bmC}{{\bm C}}
\nc{\bmD}{{\bm D}} \nc{\bmE}{{\bm E}} \nc{\bmF}{{\bm F}}
\nc{\bmG}{{\bm G}} \nc{\bmH}{{\bm H}} \nc{\bmI}{{\bm I}}
\nc{\bmJ}{{\bm J}} \nc{\bmK}{{\bm K}} \nc{\bmL}{{\bm L}}
\nc{\bmM}{{\bm M}} \nc{\bmN}{{\bm N}} \nc{\bmO}{{\bm O}}
\nc{\bmP}{{\bm P}} \nc{\bmQ}{{\bm Q}} \nc{\bmR}{{\bm R}}
\nc{\bmS}{{\bm S}} \nc{\bmT}{{\bm T}} \nc{\bmU}{{\bm U}}
\nc{\bmV}{{\bm V}} \nc{\bmW}{{\bm W}} \nc{\bmX}{{\bm X}}
\nc{\bmZ}{{\bm Z}}

\nc{\cA}{{\cal A}} \nc{\cB}{{\cal B}} \nc{\cC}{{\cal C}}
\nc{\cD}{{\cal D}} \nc{\cE}{{\cal E}} \nc{\cF}{{\cal F}}
\nc{\cG}{{\cal G}} \nc{\cH}{{\cal H}} \nc{\cI}{{\cal I}}
\nc{\cJ}{{\cal J}} \nc{\cK}{{\cal K}} \nc{\cL}{{\cal L}}
\nc{\cM}{{\cal M}} \nc{\cN}{{\cal N}} \nc{\cO}{{\cal O}}
\nc{\cP}{{\cal P}} \nc{\cQ}{{\cal Q}} \nc{\cR}{{\cal R}}
\nc{\cS}{{\cal S}} \nc{\cT}{{\cal T}} \nc{\cU}{{\cal U}}
\nc{\cV}{{\cal V}} \nc{\cW}{{\cal W}} \nc{\cX}{{\cal X}}
\nc{\cZ}{{\cal Z}}

\begin{document}

\title{Necessary Conditions on Effective Quantum Entanglement Catalysts}

\author{Yu-Min Guo}\email[]{guo\_ym@cnu.edu.cn}
\affiliation{School of Mathematical Sciences, Capital Normal University, Beijing 100048, China}

\author{Yi Shen}
\affiliation{School of Mathematical Sciences, Beihang University, Beijing 100191, China}

\author{Li-Jun Zhao}
\affiliation{School of Mathematical Sciences, Beihang University, Beijing 100191, China}

\author{Lin Chen}
\affiliation{School of Mathematical Sciences, Beihang University, Beijing 100191, China}
\affiliation{International Research Institute for Multidisciplinary Science, Beihang University, Beijing 100191, China}

\author{Mengyao Hu}
\affiliation{School of Mathematical Sciences, Beihang University, Beijing 100191, China}

\author{Zhiwei Wei }
\affiliation{School of Mathematical Sciences, Capital Normal University, Beijing 100048, China}

\author{Shao-Ming Fei}\email[]{feishm@cnu.edu.cn         (corresponding author)}
\affiliation{School of Mathematical Sciences, Capital Normal University, Beijing 100048, China}
\affiliation{Max-Planck-Institute for Mathematics in the Sciences, Leipzig 04103, Germany}

\begin{abstract}
Quantum catalytic transformations play important roles in the transformation of quantum entangled states under local operations and classical communications (LOCC). The key problems in catalytic transformations are the existence and the bounds on the catalytic states. We present the necessary conditions of catalytic states based on a set of points given by the Schmidt coefficients of the entangled source and target states. The lower bounds on the dimensions of the catalytic states are also investigated. Moreover, we give a detailed protocol of quantum mixed state transformation under entanglement assisted LOCC.
\end{abstract}

\date{\today}

\maketitle




\section{Introduction}
\label{sec:intro}
Quantum entangled states are important physical resources and play significant roles in the burgeoning field of quantum information \cite{PhysRevLett.122.190502,RevModPhys.91.025001,PhysRevA.98.042320,PhysRevLett.121.190503}. The transformation of entangled states is widely used in numerous remarkable information processing tasks such as quantum secret communication \cite{RevModPhys.81.865}, quantum key distribution \cite{BB84}, quantum super-secret coding \cite{PhysRevA.71.012316}, quantum teleportation \cite{PhysRevLett.70.1895}, quantum computing \cite{PhysRevLett.86.5188,doi:10.1098/rspa.1985.0070} and the dynamics in quantum resource theory \cite{RevModPhys.81.865}. Nevertheless, a quantum entangled state can not be transformed into another arbitrary given state under local operations and classical communications (LOCC) in general.
It is of great significance to study the conditions of such state transformations \cite{PhysRevLett.83.436}.
	
Let $\ket{\psi}_{AB}=\sum _ { i = 1 } ^ { d } \sqrt { p _ { i } }\ket{a_i}\ket{b_i}$ and $\ket{\phi}_{AB}=\sum _ { j = 1 } ^ { d } \sqrt { q _ { j } }\ket{c_i}\ket{d_i}$ be two bipartite states in Schmidt form with  the Schmidt coefficients $p_i$ and $q_j$ in the decreasing order, shared between Alice and Bob, respectively. Denote $\bp$ ($\bq$) the Schmidt vectors composed of $p_i$'s ($q_j$'s).
It is proved that $\ket{\psi}_{AB}$ can be transformed into $\ket{\phi}_{AB}$ under LOCC if and only if $\bp$ is majorized by $\bq$, namely, $\ket{\psi}\ra\ket{\phi}$ if and only if $\bp \prec\bq$,
where $\prec$ denotes the majorization relation between $\bp$ and $\bq$ \cite{Marshall1979Inequalities,Uhlmann1982book},
\bea\nonumber
\label{def:maj}
\sum _ { i = 1 } ^ { l } p _ { i }\leq \sum _ { i = 1 } ^ { l } q _ { i }, \quad \forall\, l \in \{ 1,2 , \cdots , d \},
\eea
with equality for $l = d$.
Two bipartite states $\ket{\psi}$ and $\ket{\phi}$ are said to be \textit{incomparable} if $\ket{\psi}\nrightarrow\ket{\phi}$ and $\ket{\phi}\nrightarrow\ket{\psi}$ under LOCC.
	
Different entangled states are not necessarily convertible according to quantum resource theory. Without additional resources, it is usually impossible to transform one entangled state into another one under LOCC. Jonathan and Plenio \cite{PhysRevLett.83.3566} showed that at the presence of another state called a catalyst, it is possible to transform two incomparable states from one to another under LOCC without changing the catalyst.
For incomparable $\ket{\psi}$ and $\ket{\phi}$, there could exist a bipartite state	
$\ket{\c}=\sum_{x=1}^k\sqrt{r_x}\ket{x_{A'}}\ket{x_{B'}}\in \cH^{A'B'}$ in auxiliary systems $A'$ and $B'$ such that $\ket{\psi}\ox\ket{\c}\ra\ket{\phi}\ox\ket{\c}$, namely, $\quad \bf {p}\ox\bf{r} \prec \bf {q}\ox \bf{r}$, where $\br$ is the Schmidt vector of the state $\ket{\c}$, and the state $\ket{\c}$ in this transformation process is called the catalyst state. The dimension of vector $\br$ (the Schmidt rank of $\ket{\c}$) is also called the dimension of the catalytic state.
Such state transformation assisted by an ancillary state (catalyst) is called entanglement assisted LOCC (ELOCC) transformations \cite{PhysRevLett.83.3566}.

Although the existence of catalyst states for any given transformation has been proven by using the Renyi entropies and power means \cite{Turgut_2007}, the properties of the catalysts for a given transformation are still less known. In fact, it is difficult to find a suitable catalyst for any given transformations. As mentioned in \cite{PhysRevLett.83.3566} and \cite{Sun2005The}, it seems very hard in general to solve this problem analytically, due to lack of suitable mathematical tools to deal with the majorization of tensor product vectors, especially, the identification of Schmidt descending vector of the composite state. In \cite{PhysRevA.65.042306} Bandyopadhyay and Roychowdhury presented a useful algorithm to decide the existence of catalysts. Nevertheless, to determine whether there exists a $k\times k$ catalyst for two $n\times n$ incomparable states, the algorithm runs in exponential time with complexity $O\left([(n k) !]^{2}\right)$, which shows also the difficulty of solving this problem even numerically.
The lower bound on the dimensions of a possible catalyst for a given pure entangled state transformation has been proposed in \cite{PhysRevA.79.054302}. The upper bound on the dimensions of the required catalyst is nevertheless unbounded in general \cite{ PhysRevA.64.042314}.

Catalysis is also useful to increase the maximal transformation probability in probabilistic entanglement transformation. A sufficient and necessary condition on catalysts has been derived for certain probabilistic transformations \cite{ PhysRevA.69.062310}. It has also been shown that any entangled target state can be embezzled (up to a small amount of error) by using only a family of bipartite catalysts \cite{PhysRevA.67.060302}. However, embezzling with small error requires that the catalyst Schmidt number diverges to infinity. For some mixed entangled states, there also exist some mixed catalyst states which can improve the transformation of mixed entangled states \cite{prl.85.437}.

In this paper we focus on finite-dimensional ELOCC transformations in which embezzling is not possible with a high degree of accuracy. Based on the Schmidt coefficients of the initial, target and catalyst states, we present two important bounds on the entanglement of a potential catalyst state for any pure state transformations. We analyze the ELOCC conversion between $4\times4$ entangled quantum states catalyzed by $2\times2$ catalyst, and give both necessary and sufficient conditions of catalytic states. Detailed examples are given to show that our criteria work is better than the existing ones \cite{PhysRevA.99.052348}, .

\section{Transformation of entangled states with different Schmidt ranks}
Let $\ket{\psi}_{AB}=\sum _ { i = 1 } ^ { n } \sqrt { p _ { i } }\ket{a_i}\ket{b_i}$ and $\ket{\phi}_{AB}=\sum _ { j = 1 } ^ { t } \sqrt { q _ { j } }\ket{c_i}\ket{d_i}$ be two bipartite states in Schmidt form with Schmidt ranks n and t, with the Schmidt coefficients $p_i$ and $q_j$ in decreasing order, respectively. In the following, we always add zeros to a vector with lower dimension so that it has the same dimension as the one with higher dimension.
Clearly, if $n<t$, $\ket{\psi}_{AB}$ can not be transformed into $\ket{\phi}_{AB}$ by LOCC. When $n>t$ and $p_1<\frac{1}{t}$, ${\bf p}$ is majorized by $(\frac{1}{t},\frac{1}{t},\dots,\frac{1}{t},0,\dots,0)$ and the $t$-dimensional vector $(\frac{1}{t},\frac{1}{t},\dots,\frac{1}{t})$ is majorized by any $t$-dimensional vector. Hence, there always exist an LOCC operation to accomplish the transformation $\ket{\psi}\ra\ket{\phi}$ in this case. Therefore, we have the following Lemmas.

\begin{lemma}
\label{cor1.1}
$\ket{\psi}_{AB}$ can be transformed into $\ket{\phi}_{AB}$ under LOCC if and only if $p_1<\frac{1}{t}$ for $t<n$, where $n$ is the Schmidt rank of $\ket{\psi}_{AB}$.
\end{lemma}

\begin{lemma}
\label{cor1.2}
For any $1<t<n$ and $1\leq s<n$, if
\begin{equation}
\label{1.2}
\frac{s(t-1)}{t(n-1)}<\sum_{i=n-s+1}^n p_i < \frac{s}{n},
\end{equation}
$\ket{\psi}_{AB}$ can be transformed into $\ket{\phi}_{AB}$ under LOCC.
\end{lemma}

\begin{proof}
Since $\frac{s(t-1)}{t(n-1)}<\sum_{i=n-s+1}^n p_i$, we have $\sum_{i=1}^{n-s}p_i<1-\frac{s(t-1)}{t(n-1)}$ and $p_{n-s+1}>\frac{t-1}{t(n-1)}$, which implies
\begin{equation}\nonumber
p_1\geqslant p_2\geqslant \dots \geqslant p_{n-s}\geqslant p_{n-s+1}>\frac{t-1}{t(n-1)}.
	\end{equation}
According to that
	\begin{equation}\nonumber
	p_1=1-\sum_{i=2}^{n}p_i<1-[(n-s-1)\frac{t-1}{t(n-1)}+\frac{s(t-1)}{t(n-1)}]=\frac{1}{t},
	\end{equation}
One completes the proof of Lemma by Lemma \ref{cor1.1}.
\end{proof}

Lemma 2 gives the condition on the smallest $s$ Schmidt coefficients of $\ket{\psi}_{AB}$
under LOCC state transformation. In particular, when $s=1$ we have
	\begin{equation}\nonumber
		\frac{t-1}{t(n-1)}<p_n<\frac{1}{n}.
	\end{equation}
	
The above results show that, interestingly, when one
only knows part of the information about the a
bipartite entangled quantum state, namely, a set of the
smallest Schmidt coefficients, one can determine what
entanglement transformations are allowed under LOCC.
As important channels in information processing such as quantum teleportation, the capacities of quantum entangled states are related to their dimensions. Our results show that the dimension of the target states would be reduced under LOCC transformation, which results in a reduction of the quantum channel's capacity. In the following, we consider entangled state transformation with the same dimensions and with the help of catalysts.

\section{Bounds on entanglement catalyst states}

Let $\textbf{r}$ be the Schmidt vector of a catalyst state.
For two incomparable Schmidt vector $\textbf{p},\,\textbf{q}\in \bbR^d$, if $\textbf{p}\otimes \textbf{r}\prec \textbf{q}\otimes\textbf{r}$, the following inequalities hold,
\begin{equation}
\label{cr:inc=sv}
	p_1\leqslant q_1,~~~~p_d\geqslant q_d,~~~~\sum_{i=1}^{d-1}p_i\leqslant\sum_{i=1}^{d-1}q_i.
\end{equation}

We define the set $\mathcal{L}$ to be all the values of $l$ such that $\sum_{i = 1}^{l}p_{i}>\sum_{i=1}^{l}q_{i}$,
\bea\label{L}
\mathcal { L } = \{ l \in \{ 1,2 \cdots d \} | \sum _ { i = 1 } ^ { l } (p _ { i } - q _ { i }) > 0 \}.
\eea
Denote $m_{\cL} \equiv \min\{l\big| l\in\cL\}$ and $M_{\cL} \equiv \max\{l\big| l\in\cL\}$.
From \eqref{cr:inc=sv} we have that $\{1,(d-1),\,d\}\notin\mathcal {L}$.
We call two incomparable vectors $\bp$ and $\bq$ \textsl{solvable incomparable vectors}
if $\bp$ and $\bq$ satisfy the conditions that $1,(d-1),d\notin\mathcal {L}$.
The main problem is to find the catalyst states for two solvable incomparable states, see Fig. 1.
For a fixed dimension of catalytic system, there exist two solvable incomparable vectors that
cannot be catalyzed. For instance,
$\bp=(0.414047778,0.31764445, 0.18499118, 0.083316592)$ and
$\bq=(0.428610282,0.289194489,0.212421079,0.06977415)$
are solvable incomparable vectors. But there do not exist two dimensional catalyst vectors
for $\bp$ and $\bq$ \cite{Sun2005The}.

\begin{figure}
	\centering
	\includegraphics[width=1\linewidth]{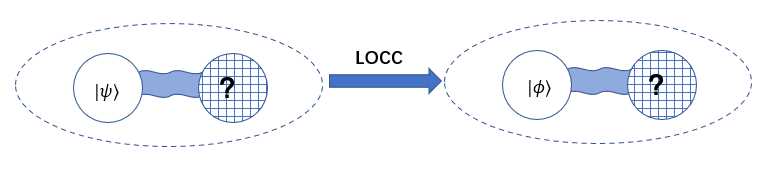}
	\caption[fig1]{Two solvable and non-comparable pure states $\ket{\psi}$ and $\ket{\phi}$
		could be transformed under LOCC by using auxiliary catalytic states.}
	\label{fig:gif1}
\end{figure}

In the following, we derive first time some general conditions on catalyst-state transformation under LOCC.

\begin{theorem}
\label{theorem2plus}
For arbitrary incomparable Schmidt vectors $\mathbf{p}, \mathbf{q}\in \mathbb{R}^{d}$, if there exists a catalyst vector $\mathbf{r} \in \mathbb{R}^{k}$, such that $\mathbf{p}\otimes\mathbf{r}\prec \mathbf{q}\otimes\mathbf{r}$, then $\mathbf{r}$ must satisfy the following inequality,
\bea\label{LBof_r1rk}
\frac{r_1}{r_k}>\max_{l\in \mathcal{L}}(\min\{\frac{p_l}{p_{l+1}},\frac{q_l}{q_{l+1}}\}).
\eea

\end{theorem}

\begin{proof}
	We prove it by contradiction. Suppose
	\bea\nonumber
	\frac{r_1}{r_k}\leqslant\min\{\frac{p_l}{p_{l+1}},\frac{q_l}{q_{l+1}}\}.
	\eea
Then $p_{l+1}r_1\leqslant p_lr_k$ and $q_{l+1}r_1\leqslant q_lr_k$, the first $lk$ largest elements of $\mathbf{p}\otimes\mathbf{r}$ and $\mathbf{q}\otimes\mathbf{r}$ are given by $\{p_ir_j\}$ and $\{q_ir_j\}$ with $1 < i < l $ and $ 1 < j < k$, respectively. Consequently one has the following relation:
	\bea\nonumber
	\begin{aligned}
		\sum_{i=1}^{kl}(\mathbf{p}\otimes\mathbf{r})_{i}^{\downarrow}&= \left(\sum_{j=1}^{l}p_{j} \right) \left(\sum_{i=1} ^{k}r_{i}\right)=\sum_{j=1}^{l}p_{j}\\
		&>\sum_{j=1}^{l}q_{j}=\left(\sum_{j=1}^{l}q_{j} \right)\left(\sum_{i=1}^{k}r_{i}\right)\\
		&=\sum_{i=1}^{kl}(\mathbf{q}\otimes\mathbf{r})_{i}^{\downarrow},
	\end{aligned}
	\eea
where $(\mathbf{x})_{i}^{\downarrow}$ means that the components of the vector $\mathbf{x}$ are arranged in decreasing order, which contradicts the condition $\mathbf{p}\otimes\mathbf{r}\prec\mathbf{q}\otimes\mathbf{r}$.
\end{proof}

The inequality \eqref{LBof_r1rk} has another equivalent expression,
\bea\nonumber
\frac{r_k}{r_1}<\min_{l\in \mathcal{L}}(\max\{\frac{p_{l+1}}{p_l},\frac{q_{l+1}}{q_l}\}).
\eea
Moreover, due to the definitions of $m_{\cL}$ and $M_{\cL}$, we have
\bea\nonumber
\begin{aligned}
		\frac { r _ { 1 } } { r _ { k } } &> \min \{ \frac { p _ { m_{\cL} } } { p _ { m_{\cL} + 1 } } , \frac { q _ { m_{\cL} } } { q _ { m_{\cL} + 1 } } \},\\[1mm]
		\frac { r _ { 1 } } { r _ { k } } &> \min \{ \frac { p _ { M_{\cL} } } { p _ { M_{\cL} + 1 } } , \frac { q _ { M_{\cL} } } { q _ { M_{\cL} + 1 } } \}.
\end{aligned}
\eea

Let $l_1,l_2,\dots,l_m$; $l_i<l_{i+1}$ be the elements of the set $\mathcal {L}$ defined in (\ref{L}), where $m=\#\mathcal { L }$ is the number of the elements in $\mathcal { L }$.

\begin{theorem}
\label{theorem-main1}
Let $\mathbf { p }, \mathbf { q } \in \mathbb { R } ^ { d }$ be arbitrary two solvable incomparable Schmidt vectors. If $\mathbf { p } \otimes \mathbf { r } \prec \mathbf { q } \otimes \mathbf { r }$, we have
\begin{equation}\nonumber
\min_{\tau=1,\dots,m+1}\{r_{j_{\tau}}q_{l_{\tau}}\}<\max_{\tau=1,\dots,m+1}\{r_{j_{\tau}+1}q_{l_{\tau-1}+1}\},
\end{equation}
where $l_0=0$, $l_{m+1}=d$, and all the terms with $r_0$ or $r_{k+1}$ are neglected.
The indices $\{j_s\}_{s=1}^{m+1}$ satisfy $0\leqslant j_i\leqslant k$ for $i=1,2,\dots,m+1$,
$ j_{i+1}\leqslant j_{i}$ for $i=1,2,\dots,m$, and $j_{i+1}< j_{i}$ for one $i=1,2,\dots,m+1$ at least.
\end{theorem}

\begin{proof}
 We prove it by contradiction. Denote $\Xi=\sum_{i=1}^{m}(j_i-j_{i+1})l_i+j_{m+1}d$. If 	
$$
\min_{\tau=1,\dots,m+1}\{r_{j_{\tau}}q_{l_{\tau}}\}\geqslant
\max_{\tau=1,\dots,m+1}\{r_{j_{\tau}+1}q_{l_{\tau-1}+1}\},
$$
then the $\Xi$ largest elements of $\mathbf{q}\otimes\mathbf{r}$ are $r_1q_1$,  $r_1q_2$, $\dots$, $r_1q_d$, $r_2q_1$, $r_2q_2$, $\dots$, $r_2q_d$, $\dots$, $r_{j_{m+1}}q_d$, $r_{j_{m+1}+1}q_1$, $r_{j_{m+1}+1}q_2$, $\dots$, $r_{j_{m+1}+1}q_{l_m}$, $r_{j_{m+1}+2}q_1$, $\dots$, $r_{j_{m+1}+2}q_{l_m}$,$\dots$, $r_{j_m}q_{l_m}$, $r_{j_m+1}q_{l_1}$, $\dots$, $r_{j_m+1}q_{l_{m_1}}$, $\dots$, $r_{j_1}q_1$, $\dots$, $r_{j_1}q_{l_1}$.
By the definition of $\mathcal { L }$, one has for any $\tau$, $\sum_{i=1}^{l_{\tau}}p_i>\sum_{i=1}^{l_{\tau}}q_i$.

Since
\begin{equation}\nonumber
\begin{aligned}
 \sum_{s=1}^{\Xi}(\mathbf{q}\otimes\mathbf { r } ) _ { s } ^{\downarrow}&=\sum_{\tau=1}^{m+1}(\sum_{u=j_{\tau+1}+1}^{j_{\tau}}r_u)\cdot(\sum_{v=1}^{l_{\tau}}q_v)\\
 &<\sum_{\tau=1}^{m+1}(\sum_{u=j_{\tau+1}+1}^{j_{\tau}}r_u)\cdot(\sum_{v=1}^{l_{\tau}}p_v)\\
 &\leqslant\sum_{s=1}^{\Xi}(\mathbf{p}\otimes\mathbf { r } ) _ { s } ^{\downarrow},
\end{aligned}
\end{equation}
where $j_{m+2}=0$, the sum of the largest $\Xi$ elements of $\mathbf {q}\otimes\mathbf {r}$ is strictly less than the sum of the largest $\Xi$ elements of $\mathbf {p}\otimes\mathbf {r}$.
\end{proof}

{\it Remark 1} Theorem \ref{theorem-main1} has another useful equivalent statement: Let $\mathbf { p }$ and $\mathbf { q } \in \mathbb { R } ^ { d }$
be two fixed solvable incomparable Schmidt vectors, and $\mathbf { r } \in \mathbb { R } ^ { k }$ the Schmidt vector of a catalyst state. Denote $\{j_s\}_{s=1}^{m+1}$ a sequence of natural numbers satisfying $0< j_{m'+1}< j_{m'}<\dots< j_1< k$, where $m'\leqslant m=\#\mathcal { L }$, $ m$ is the number of elements in $\mathcal { L }$. Consider a subset $\bar{\mathcal { L }}$ of $\mathcal { L }$ with $ m'$ elements,
$\bar{\mathcal { L }}=\{l_1,l_2,\dots,l_{m'}|l_i\in\mathcal { L },~~l_i<l_{i+1}\}$. Set $l_0=0$ and $l_{m+1}=d$. If $\mathbf { p } \otimes \mathbf { r } \prec \mathbf { q } \otimes \mathbf { r }$, then
\begin{equation}\nonumber
\min_{\tau=1,\dots,m',m'+1}\{r_{j_{\tau}}q_{l_{\tau}}\}
<\max_{\tau=1,\dots,m',m'+1}\{r_{j_{\tau}+1}q_{l_{\tau-1}+1}\}.
\end{equation}

To elucidate their practical applications in actual catalytic processes, we enumerate several special corollaries of the Theorem \ref{theorem-main1} below.

We consider the case that $\{j_i\}_{i=1}^{m+1}$ take two different values, say, $0\leqslant c_1=j_{m+1}=\dots=j_s<j_{s-1}=\dots=j_1=c_2\leqslant k $. In this case, Theorem \ref{theorem-main1} gives rise to the following simple form:

\begin{corollary}
	\label{cor1}
If $\mathbf { p } \otimes \mathbf { r } \prec \mathbf { q } \otimes \mathbf { r }$, then
	\begin{equation}\nonumber
	\min\{r_{c_1}q_d,r_{c_2}q_{l_s}\}<\max\{r_{c_1+1}q_{l_{s}+1},r_{c_2+1}q_1\},
	\end{equation}
where $l_0=0$, $l_{m+1}=d$, $1\leq c_1<c_2\leq k$, and the terms with $r_{k+1}$ are ignored.
\end{corollary}

Corollary \ref{cor1} still gives rise to a series of conditions one catalyst. In particular, for $c_2=c_1+1=s$, we have the following conclusions, see proof in Appendix A:

\begin{proposition}
\label{remark-c2=c1+1}
If $\mathbf { p } \otimes \mathbf { r } \prec \mathbf { q } \otimes \mathbf { r }$, we have
\begin{equation}
\label{three-inequalities}
\frac{q_1}{q_d}>\frac{r_{s-1}}{r_{s+1}},~~or~~ \frac{q_1}{q_l}>\frac{r_s}{r_{s+1}}~~ or ~~\frac{q_{l+1}}{q_d}>\frac{r_{s-1}}{r_s}
\end{equation}
for any $l\in\mathcal{ L }$, where $s=2,3,\cdots,k-1$. If $s=1$ or $s=k$, then
\begin{equation}\label{new_d1}
\frac{q_d}{q_{l+1}}<\frac{r_k}{r_{k-1}},~~~
\frac{q_1}{q_l}>\frac{r_1}{r_2}.
\end{equation}
\end{proposition}

The above Proposition gives detailed criteria on catalytic vectors. In particular, it is often difficult to determine the values of all components of the catalytic vectors for high dimensional case.
The inequality (\ref{new_d1}) says that one may first simply verify the first two and the last two components of the catalytic vectors, which could greatly simplify the construction of qualified catalytic states.

Recently, by using a very nice idea the authors in \cite{PhysRevA.99.052348} derived elegant results on the minimum and maximum bounds about the auxiliary catalysts under ELOCC. It is showed that (Theorem 1 in \cite{PhysRevA.99.052348}) if $\mathbf{p} \otimes \mathbf{r} \prec \mathbf{q} \otimes \mathbf{r}$, $\mathbf{r}$ satisfies
\begin{equation}\label{PRA99}
\begin{array}{rcl}
\displaystyle\max_{v \in(1,2, \ldots, k-1)}(\frac{r_{v}}{r_{v+1}})&<&\displaystyle\min (\frac{q_{1}}{q_{m_{\cL}}}, \frac{q_{M_{\cL}+1}}{q_{d}}),\\[4mm]
\displaystyle\frac{r_{1}}{r_{k}}&>&\displaystyle\max _{l \in \mathcal{L}}(\frac{q_{l}}{q_{l+1}}).
\end{array}
\end{equation}
As the Proposition \ref{remark-c2=c1+1} takes into account additionally two adjacent components of the vector $\mathbf{r}$, it gives rise to finer constraints on
the catalyst vector. The proposition \ref{remark-c2=c1+1} is generally better than \eqref{PRA99} since, instead of that all the inequalities are simultaneously true,
\eqref{three-inequalities} only requires that at least one of the inequalities is true, see also the remark in Appendix B.

Furthermore, for a special case of corollary \ref{cor1} we have, see proof in Appendix C,
\begin{corollary}
\label{coro-dual}
Let $c_1=k-1$, $c_2=k$ and $s=m$, we have
	\begin{align}
	\label{equ:II}
	\frac{q_d}{q_{M_{\cL}+1}}&<\frac{r_k}{r_{k-1}}
	\end{align}
\end{corollary}
if $\mathbf { p } \otimes \mathbf { r } \prec \mathbf { q } \otimes \mathbf { r }$.

{\it Remark 2} If we take $c_1=0$, $c_2=1$ and $s=1$, then we get another special case of corollary \ref{cor1},
\begin{align}
\label{equ:I}
\frac{q_1}{q_{m_{\cL}}}&>\frac{r_1}{r_2}.
\end{align}

The inequality \eqref{equ:II} can be also written in a symmetric form:
\bea\nonumber
\frac{r_k}{r_{k-1}}>\min\{\frac{p_d}{p_{M_{\cL}+1}},\frac{q_d}{q_{M_{\cL}+1}}\}.
\eea

Corollary \ref{coro-dual} can be used, for example, to test whether a 2-dimensional catalytic state can be used to catalyze a given pair of solvable incomparably entangled states. We can also the corollary \ref{cor1} to verify the availability of higher dimensional catalytic states. For the case of 3-dimensional catalytic states, we can directly check whether each component of the Schmidt vector satisfies the condition given in corollary \ref{cor1}. For higher dimensional case, we can take any 3 components of the Schmidt vectors to test whether they violate the above corollary.

The following example illustrates that our result (\ref{equ:II}) behaviors better than the result (\ref{equ:I}) from \cite{PhysRevA.99.052348} in determining whether a catalytic state can complete a certain catalytic process.

\begin{example}
Consider two bipartite incomparable pure states with the corresponding Schmidt vectors $\mathbf{p}=(0.4, 0.35, 0.15, 0.1)$ and $\mathbf{q}=(0.5, 0.2, 0.2, 0.1)$, respectively.
Obviously , these two states can not transformed into each other under LOCC.
Now consider a two-dimensional catalyst state with Schmidt vector $ \mathbf { r }=(0.7, 0.3)$.
(\ref{equ:II}) says that if $r_1/r_2\geqslant q_3/q_4$, then $\mathbf { p } \otimes \mathbf { r } \nprec \mathbf { q } \otimes \mathbf { r }$. For this example, one has
$q_1/q_2=2.5>r_1/r_2=2.\dot{3}> q_3/q_4=2$, where $\dot{3}$ stands for repeated decimal. Simple calculation gives rise to $\mathbf { p } \otimes \mathbf { r } \nprec \mathbf { q } \otimes \mathbf { r }$.
\end{example}

\section{bounds on the dimension of catalyst}
According to the resource theory, the dimensions of catalytic states introduced on auxiliary systems should be as small as possible. For fixed incomparable Schmidt vectors $\mathbf{p}$ and $\mathbf{q}\in \mathbb{R}^{d}$, it is of significance to find the minimum dimension of their catalyst states. It has been proved in \cite{PhysRevA.64.042314} that all non-uniform vectors are potentially useful as catalyst states. Y. R. Sanders and G. Gour give a lower bound on the dimensions of catalyst states based on the $k$-th concurrence $C_{d-1}$ \cite{PhysRevA.79.054302}, which works for the case that $C_{d-1}(\ket{\psi})<C_{d-1}(\ket{\phi})$. However, generally for large Schmidt dimensions of $\ket{\psi}$ and $\ket{\phi}$, the $k$-th concurrences are often difficult to calculate.

Next, we analyze the relationship between the dimensions of the auxiliary systems and the set $\mathcal{L}$. As the minimum dimension of the catalytic state is 2, we can derive the following corollary by Theorem \ref{theorem-main1}.

\begin{corollary}
Let $\bp$ and $\bq$ be two solvable incomparable Schmidt vectors with dimension $n$, associated with the states $\ket{\psi}$ and $\ket{\varphi}$, respectively. Denote $\cL=\{l_1,l_2,\dots,l_m\}$.
Let $0<r<0.5$ be a real number. If $q_{l_i+1}r<q_m(1-r)$ for any integer $1\leqslant i\leqslant m$, and
$q_{l_i+1}r<q_{l_j+1}(1-r)$, $q_{l_j+1}(1-r)<q_{l_i}r$ for any integers $i,j:1\leqslant i<j\leqslant m$,	 then the two-dimensional vector $\br=(1-r,r)$ cannot make $\bp\otimes\br\prec \bq\otimes\br$ true. Namely, the catalytic state $\ket{\chi}$ with Schmidt vector $\br=(1-r,r)$ cannot catalyze the states $\ket{\psi}$ and $\ket{\varphi}$ under LOCC.
\end{corollary}

In other words, the dimension of catalytic state (if it exists) for such $\ket{\psi}$ and $\ket{\varphi}$, as mentioned in the above corollary, must be strictly greater than two. For arbitrary dimensions of the catalytic states, we have

\begin{corollary}
\label{cor-B2}
For any pair of solvable incomparable Schmidt vectors $\bp$ and $\bq$ with dimension $n$, associated with the states $\ket{\psi}$ and $\ket{\varphi}$, respectively, if any $k$-dimensional probability vector $\br=(r_1,r_2,\dots,r_k)$ does not satisfy the conditions in Theorem \ref{theorem-main1}, $\ket{\psi}$ and  $\ket{\varphi}$ cannot be catalyzed by a state with $k$-dimensional Schmidt vector.
\end{corollary}

The above corollary helps to reduce the difficulty of finding the minimum dimensions for catalytic system. In addition, we can utilize the Proposition \ref{remark-c2=c1+1} to give a lower bound on the dimension of catalyst states for arbitrary solvable incomparable Schmidt vectors  $\bp$ and $\bq$. Denote
\begin{align*}
 a&\equiv \max (\frac{q_{1}}{q_{m_{\mathcal{L}}}}, \frac{q_{M_{\mathcal{L}}+1}}{q_{d}}),\\
 b&\equiv \max _{l \in \mathcal{L}}(\frac{q_{l}}{q_{l+1}})\geqslant 1, \\
 c&\equiv\max_{l\in\mathcal{L}}(\min\{\frac{p_l}{p_{l+1}},\frac{q_l}{q_{l+1}}\}).
\end{align*}
\begin{theorem}
\label{theorem-bound}
For any two solvable incomparable Schmidt vectors $\mathbf { p }, \mathbf { q } \in \mathbb { R } ^ { d }$ and for any Schmidt vector $\mathbf { r } \in \mathbb { R } ^ { k }$, if $\mathbf { p } \otimes \mathbf { r } \prec \mathbf { q } \otimes \mathbf { r }$, then
the dimension $k$ of $\mathbf { r }$ must satisfy
	\begin{equation}
		\label{bound}
		k>\frac{\ln (c)}{\ln (a\sqrt{b})}+1.
	\end{equation}
\end{theorem}
\begin{proof}
For arbitrary $r_s$, $s=2,3,\dots,k-1$, from Proposition \ref{remark-c2=c1+1} at least one of the inequalities in \eqref{three-inequalities} holds. If $\frac{q_1}{q_d}>\frac{r_{s-1}}{r_{s+1}}$ holds for some $s$, rewriting $\frac{q_1}{q_d}=\frac{q_1}{q_l}\frac{q_{l+1}}{q_d}\frac{q_l}{q_{l+1}}$, we have $\frac{r_{s-1}}{r_{s+1}}<a^2b$. Otherwise,
\begin{align*}
\frac{r_{s}}{r_{s+1}}&<\max\{\min_{l\in \mathcal{L}}(\frac{q_1}{q_l}),\min_{l\in \mathcal{L}}(\frac{q_{l+1}}{q_d})\} \\
&=\max (\frac{q_{1}}{q_{m_{\mathcal{L}}}}, \frac{q_{M_{\mathcal{L}}+1}}{q_{d}})\leqslant a\sqrt{b}.
\end{align*}
Taking into account \eqref{new_d1}, we get
	\begin{equation}
	\label{lowbound1}
	\frac{r_{s}}{r_{s+1}}<a\sqrt{b},~~ \forall s=1,2,\dots,k-1.
	\end{equation}
Therefore, $\frac{r_{1}}{r_{k}}=\prod_{v=1}^{k-1} \frac{r_{v}}{r_{v+1}}<(a\sqrt{b})^{k-1}$.
On the other hand, we have from Theorem \ref{theorem2plus} that $\frac{r_{1}}{r_{k}}>c$.
Hence,
\begin{equation}\nonumber
c<\frac{r_{1}}{r_{k}}<(a\sqrt{b})^{k-1}.
\end{equation}
The above relations give rise to (\ref{bound}).
\end{proof}

The bound \eqref{bound} depends only on the parameters $a$, $b$ and $c$, which are exclusively determined by the Schmidt vectors $\mathbf {p}$ and $\mathbf {q}$ and can be easily calculated.

\section{maximum probability of transformation and majorization distance}
\label{sec:pd}

For any two incomparable states $\ket{\psi}$ and $\ket{\phi}$, the maximum probability of transforming $\ket{\psi}$ to $\ket{\phi}$ under local operations is given by \cite{PhysRevLett.83.1046}
\beq
\nonumber
P_{\max}'(\ket{\psi}\ra\ket{\phi})=\min_{l\in\{1,2,\cdots,d\}} \frac{E_{l}(\ket{\psi})}{E_{l}(\ket{\phi})},
\eeq
where $E_{l}(\ket{\psi})=\sum_{i=l}^d p_i$.

To characterize the transformation of two incomparable states $\ket{\psi}$ and $\ket{\phi}$ assisted with a catalytic state $\chi$ of Schmidt vector $\br =(r_1, r_2, \cdots, r_k)$, we consider the modified maximum probability $P _ { \max } ( | \psi \rangle \otimes | \chi \rangle \rightarrow | \phi \rangle \otimes | \chi \rangle )$ \cite{PhysRevLett.83.1046},
\bea\nonumber
\begin{aligned}
&P _ { \max } ( | \psi \rangle \otimes | \chi \rangle \rightarrow | \phi \rangle \otimes | \chi \rangle ) \\&= \min _ { l ^ { \prime } \in \{ 1,2 , \cdots , k d \} } \frac { E _ { l ^ { \prime } } ( | \psi \rangle \otimes | \chi \rangle ) } { E _ { l ^ { \prime } }  ( | \phi \rangle \otimes | \chi \rangle ) }\end{aligned}
\eea
where $ E _ { l ^ { \prime } } ( | \psi \rangle \otimes | \chi \rangle ) = 1 - \sum _ { i = 1 } ^ { l ^ { \prime } - 1 } ( \mathbf { p } \otimes \mathbf { r } ) _ { i } ^ {\downarrow }$.

\begin{proposition} Let $\delta ( | \psi \rangle \otimes | \chi \rangle \rightarrow | \phi \rangle \otimes | \chi \rangle )$ be the modified majorization distance \cite{Horodecki_2018} between the product states
$| \psi \rangle \otimes | \chi \rangle$ and $| \phi \rangle \otimes | \chi \rangle$,
\bea\nonumber
\begin{aligned}
&\delta ( | \psi \rangle \otimes | \chi \rangle \rightarrow | \phi \rangle \otimes | \chi \rangle )\\ &= 2 \max _ { l ^ { \prime } \in \{ 1 , \cdots , k d \} } \sum _ { i = 1 } ^ { l ^ { \prime } } \left( ( \mathbf { p } \otimes \mathbf { r } ) _ { i } ^ { \downarrow } - ( \mathbf { q } \otimes \mathbf { r } ) _ { i } ^ { \downarrow } \right).
\end{aligned}
\eea
We have $P_{\max}=1$ if and only if $\delta=0$.
\end{proposition}

\begin{proof}
	We have
	\bea\nonumber
	\begin{aligned}
		\delta=0& \Leftrightarrow \max _ { l ^ { \prime } \in \{ 1 , \cdots , k d \} } \sum _ { i = 1 } ^ { l ^ { \prime } } \left( ( \mathbf { p } \otimes \mathbf { r } ) _ { i } ^ { \downarrow } - ( \mathbf { q } \otimes \mathbf { r } ) _ { i } ^ { \downarrow } \right)=0 \\
		& \Leftrightarrow \forall l ^ { \prime },\sum _ { i = 1 } ^ { l ^ { \prime } }  ( \mathbf { p } \otimes \mathbf { r } ) _ { i } ^ { \downarrow } \leqslant \sum _ { i = 1 } ^ { l ^ { \prime } } ( \mathbf { q } \otimes \mathbf { r } ) _ { i } ^ { \downarrow }  \\
		& \Leftrightarrow \mathbf { p } \otimes \mathbf { r } \prec \mathbf { q } \otimes \mathbf { r }.
	\end{aligned}
	\eea
Then we have
	\bea\nonumber
	\begin{aligned}
		P_{\max}=1& \Leftrightarrow \forall l ^ { \prime },~  E _ { l ^ { \prime } } ( | \psi \rangle \otimes | \chi \rangle )\geqslant  E _ { l ^ { \prime } }  ( | \phi \rangle \otimes | \chi \rangle ) \\
		& \Leftrightarrow  \sum _ { i = 1 } ^ { l ^ { \prime } - 1 } ( \mathbf { p } \otimes \mathbf { r } ) _ { i } ^ { \downarrow }\leqslant  \sum _ { i = 1 } ^ { l ^ { \prime } - 1 } ( \mathbf { q } \otimes \mathbf { r } ) _ { i } ^ { \downarrow } \\
		& \Leftrightarrow \mathbf { p } \otimes \mathbf { r } \prec \mathbf { q } \otimes \mathbf { r }.
	\end{aligned}
	\eea
To sum up, we have $\delta=0\Leftrightarrow \mu(| \psi \rangle \otimes | \chi \rangle) \prec \mu(| \phi \rangle \otimes | \chi \rangle)\Leftrightarrow P_{\max}=1$.
\end{proof}

\section{Protocols for mixed state transformation assisted by catalysts}

We consider the LOCC transformations for mixed states with catalysts.
Let $\r_A=\sum_ip_i\proj{i}$, $\s_A=\sum_j q_j\proj{j}$ and $\a_{A'}$ be the reduced density matrices of the states $| \psi ^ { A B } \rangle$, $| \phi ^ { A B } \rangle $ and $\ket{\chi}_{A'B'}$, respectively. By matrix majorization relations, we have $\r_A\otimes\a_{A'}\ra\s_A\otimes\a_{A'}$ if and only if $\textbf{p}\otimes \textbf{r}\prec \bq\otimes\br$.

Consider a mixed state $\rho_s$ with ensemble $\{p_i,\ket{\psi_i}\}_{i=1}^m$, and a catalyst state $\ket{\chi}=\sqrt{1-t}\ket{0}+\sqrt{t}\ket{1}$ with Schmidt rank $2$. Assume that $\ket{\psi_i}$ and $\ket{\varphi}$ have the same Schmidt rank, and there exists common $t$ such that $\ket{\psi_i}\otimes\ket{\chi}\longrightarrow\ket{\varphi}\otimes\ket{\chi}$ holds. Namely, there exists
LOCC operation $U_i^{LOCC}$ such that $U_i^{LOCC}(\ket{\psi_i,\chi})=\ket{\varphi,\chi}$ for all $i=1,...,m$. Then the following protocol transforms the state $\rho_s\otimes\ket{\chi}$ to $\ket{\varphi}\otimes\ket{\chi}$ under LOCC:

1) Introducing an ancillary $m$-dimensional system $\cH_a$ with orthonormal basis
    $\{\ket{i}|i=1,2,\dots,m\}$. Prepare the initial state of the system, catalyst and ancillary to be
$$
\rho_{sca}=\sum_{i=1}^m p_i\ket{\psi_i}\bra{\psi_i}\otimes\ket{\chi}\bra{\chi}\otimes\ket{i}\bra{i}.
$$

2) Constructing a quantum channel $\Lambda_{sca}$ on the initial state,
$\Lambda_{sca}(\cdotp)=\sum_{ i=1}^mK_i(\cdotp)K_i^{\dagger}$, with the Kraus operators given by
$K_i=U_i^{LOCC}\otimes\ket{i}\bra{i}$ satisfying $\sum_{ i=1 }^mK_i^{\dagger}K_i=\bf{I}$.

3) Applying the quantum channel $\Lambda_{sca}$, we obtain	
	\begin{widetext}
	\bea\nonumber
	\begin{aligned}
		\Lambda_{sca}(\rho_{sca})&=\sum_{ i=1}^mK_i(\rho_{sca})K_i^{\dagger}\\
		&=\sum_{ j=1}^mU_j^{LOCC}\otimes\ket{j}\bra{j}(\sum_{ i = 1 }^mp_i\ket{\psi_i,\chi,i}\bra{\psi_i,\chi,i})(U_j^{LOCC})^{\dagger}\otimes\ket{j}\bra{j}\\
		&=\sum_{ i = 1 }^mp_iU_i^{LOCC}(\ket{\psi_i,\chi}\bra{\psi_i,\chi})(U_i^{LOCC})^{\dagger}\otimes\ket{i}\bra{i}\\
		&=\sum_{ i=1 }^mp_i\ket{\varphi,\chi}\bra{\varphi,\chi}\otimes\ket{i}\bra{i}\\
		&=\ket{\varphi,\chi}\bra{\varphi,\chi}\otimes\sum_{i=1}^mp_i\ket{i}\bra{i},
	\end{aligned}
	\eea
	\end{widetext}
which realizes the conversion from a mixed state $\sum_{ i }p_i\ket{\psi_i}\bra{\psi_i}$ to a pure state $\ket{\varphi}$ assisted by a catalyst state $\ket{\chi}$ and the ancillary system under LOCC quantum channel. Tracing over the catalyst and the ancillary systems, one gets the pure state $\ket{\varphi}$.

\section{concluding discussions}
We have investigated the bounds on catalytic states and presented general necessary conditions for the existence of catalytic states. To the open question whether any two solvable non-comparably entangled pure states have a catalytic state which catalyzes the transformation, our intuition is that the existence of the catalytic states is related to the set $\mathcal { L } = \{ l \in \{ 1,2 \cdots d \} | \sum _ { i = 1 } ^ { l } (p _ { i } - q _ { i }) > 0 \}$ given by the Schmidt coefficients of the two states. Our bounds on catalytic states give effective ways in entanglement catalyst for not only pure but also mixed states. These results on catalysts would deepen our understanding on quantum entanglement transformations under LOCC and highlight further the related investigations.

\section*{Acknowledgments}
YMG and SMF acknowledges the support from NSF of China under Grant No. 12075159, Beijing Natural Science Foundation (Z190005), Academy for Multidisciplinary Studies, Capital Normal University, and Shenzhen Institute for Quantum Science and Engineering, Southern University
of Science and Technology, Shenzhen 518055, China (No. SIQSE202001).
YS, MYH, LJZ and LC were supported by the NNSF of China (Grant Nos. 11871089 and 11947241), and the Fundamental Research Funds for the Central Universities (Grant Nos. KG12080401 and ZG216S1902).

\section*{Appendix}

\subsection{Proof of Proposition \ref{remark-c2=c1+1}}
\begin{proof}
If the conclusion is not true, one would have
$q_1r_{s+1}\leq q_lr_s$, $q_{l+1}r_s\leq q_dr_{s-1}$ or $q_1r_{s+1}\leq q_dr_{s-1}$, which imply that
 the largest $d(s-1)+l$ elements of $\mathbf { q } \otimes \mathbf { r }$ are $q_1r_1,\cdots,q_dr_1$, $q_1r_2,\cdots,q_dr_2$, $\cdots$, $q_1r_s,\cdots,q_lr_s$. The summation of these largest $d(s-1)+l$ elements of $\mathbf { q } \otimes \mathbf { r }$ gives
\begin{align*}
\sum_{ i = 1 }^{d(s-1)+l}( \mathbf { q } \otimes \mathbf { r } ) _ { i }&=\sum_{ i = 1}^d\sum_{j=1}^{s-1}q_ir_j+\sum_{t=1}^lq_tr_s\\
&=\sum_{ i = 1 }^{s-1}r_i+r_s\sum_{j= 1}^lq_j.
\end{align*}
Meanwhile, the largest $d(s-1)+l$ elements of $\mathbf { p } \otimes \mathbf { r }$ are not less than the sum of the terms $p_1r_1,\cdots,p_dr_1$, $p_1r_2,\cdots,p_dr_2$, $\cdots$, $p_1r_s,\cdots,p_lr_s$,
\begin{align*}
\sum_{ i = 1 }^{d(s-1)+l}( \mathbf { p } \otimes \mathbf { r } ) _ { i } ^ { \downarrow }&\geq\sum_{ i = 1}^d\sum_{j=1}^{s-1}p_ir_j+\sum_{t=1}^lp_tr_s\\
&=\sum_{ i = 1 }^{s-1}r_i+r_s\sum_{j= 1}^lp_j.
\end{align*}

On the other hand, since $l\in\mathcal{ L }$ and $\sum_{j= 1}^lp_j >\sum_{j= 1}^lq_j $, we have
\begin{align*}
\sum_{ i = 1 }^{d(s-1)+l}( \mathbf { p } \otimes \mathbf { r } ) _ { i } ^ { \downarrow }\geqslant&\sum_{ i = 1 }^{s-1}r_i+r_s\sum_{j= 1}^lp_j \\
	>&\sum_{ i = 1 }^{s-1}r_i+r_s\sum_{j= 1}^lq_j\\
	=&\sum_{ i = 1 }^{d(s-1)+l}( \mathbf { q } \otimes \mathbf { r } ) _ { i } ^ { \downarrow }
\end{align*}
for any $s=2,3,\cdots,k-1$, which contradicts the assumption that $\mathbf { p } \otimes \mathbf { r } \prec \mathbf { q } \otimes \mathbf { r }$.

Next we prove the first inequality in \eqref{new_d1}. Since $l\in\mathcal{D}_1$, we have
$p_1+\cdots+p_l>q_1+\cdots+q_l$ and $p_1+\cdots+p_{l+1}\leq q_1+\cdots+q_{l+1}$, which give rise to
$p_{l+1}<q_{l+1}$. Taking into account $p_d>q_d$, we have $\frac{p_d}{p_{l+1}}>\frac{q_d}{q_{l+1}}$
as $\mathbf { p }$ and $\mathbf { q }$ are solvable incomparable.
	
We use proof by contradiction. If $\frac{q_d}{q_{l+1}}\geq\frac{r_k}{r_{k-1}}$, i.e., $q_dr_{k-1}\geq q_{l+1}r_k$, then the smallest $d-l$ elements of $\mathbf { q } \otimes \mathbf { r }$ are $q_{l+1}r_k,q_{l+2}r_k,\cdots,q_dr_k$. Moreover,
the sum of the smallest $d-l$ elements of $\mathbf { p } \otimes \mathbf { r }$ is not larger than the sum of $p_{l+1}r_k,p_{l+2}r_k,\cdots,p_dr_k$. Because of $l\in\mathcal{L}$, we get $p_{l+1}+\cdots+p_d<q_{l+1}+\cdots+q_d$. Therefore, we have the following relation,
	\begin{eqnarray*}
		\sum _ { i = dk-d+l+1 } ^ { dk } ( \mathbf { p } \otimes \mathbf { r } ) _ { i } ^ { \downarrow }\leqslant\sum_{ i = l+1 }^dp_ir_k\\
		<\sum_{ i = l+1 }^dq_ir_k=	 \sum _ { i = dk-d+l+1 } ^ { dk} ( \mathbf { q } \otimes \mathbf { r } ) _ { i } ^ { \downarrow },
	\end{eqnarray*}
which contradicts the assumption that $\mathbf { p } \otimes \mathbf { r } \prec \mathbf { q } \otimes \mathbf { r }$, and proves the first inequality in \eqref{new_d1}.
	
Concerning the second inequality in \eqref{new_d1}, we use again the proof by contradiction. If $\frac{q_1}{q_l}\leq\frac{r_1}{r_2}$, the largest $l$ elements of $\mathbf { q } \otimes \mathbf { r }$ are $q_1r_1,q_2r_1,\cdots,q_lr_1$. Since the sum of the largest $l$ elements of $\mathbf { p } \otimes \mathbf { r }$ is not lees than the sum of $p_1r_1,p_2r_1,\cdots,p_lr_1$ for $l\in\mathcal{L}$, we have $p_1+\cdots+p_l>q_1+\cdots+q_l$. Then we have the relation,
	\begin{eqnarray*}
		\sum _ { i = 1 } ^l( \mathbf { p } \otimes \mathbf { r } ) _ { i } ^ { \downarrow }\geqslant\sum_{ i = 1 }^lp_ir_1\\
		>\sum_{ i = 1 }^lq_ir_1=	 \sum _ { i = 1} ^ l ( \mathbf { q } \otimes \mathbf { r } ) _ { i } ^ { \downarrow },
	\end{eqnarray*}
which contradicts the assumption that $\mathbf { p } \otimes \mathbf { r } \prec \mathbf { q } \otimes \mathbf { r }$ and proves the second inequality in \eqref{new_d1}.
\end{proof}

\subsection{A remark on the relationship between Proposition \ref{remark-c2=c1+1} and the Theorem 1 in \cite{PhysRevA.99.052348}}
In the proof of Theorem 1 in \cite{PhysRevA.99.052348}, it is mentioned that ``$q_{m} r_{v^{\prime}} \geqslant q_{1} r_{v^{\prime}+1}$ implies that the first $\left(v^{\prime}-1\right) d+m$ elements of $(\mathbf{q} \otimes \mathbf{r})^{\downarrow}$ consist of $q_{x} r_{y}$, $\forall x \in\{1,2, \ldots, d\}$ and $\forall y \in\left\{1,2, \ldots, v^{\prime}-1\right\}$, along with $q_{x^{\prime}} r_{v^{\prime}}$ $\forall x^{\prime} \in\{1,2, \ldots, m\}$". Generally this could be not true because some elements of $q_{x} r_{y}$, $x \in\{1,2, \ldots, d\}$, $y \in\left\{1,2, \ldots, v^{\prime}-1\right\}$, along with $q_{x^{\prime}} r_{v^{\prime}}$ $\forall x^{\prime} \in\{1,2, \ldots, m\}$ might less than some elements of $q_{x^{\prime}} r_{v^{\prime}}$, $x^{\prime} \in\{1,2, \ldots, d\}$, $y^{\prime} \in\{ v^{\prime}+1, \ldots, k\}$, along with $q_{x^{\prime}} r_{v^{\prime}}$ $\forall x^{\prime} \in\{m+1, \ldots, d\}$, when $ q_{d} r_{v^{\prime}-1} < q_{m+1} r_{v^{\prime}}$ or $ q_{d} r_{v^{\prime}-1} < q_{1} r_{v^{\prime}+1}$. This implies that the first $\left(v^{\prime}-1\right) d+m$ elements of $(\mathbf{q} \otimes \mathbf{r})^{\downarrow}$ do not consist of $q_{x} r_{y}$, $\forall x \in\{1,2, \ldots, d\}$ and $\forall y \in\left\{1,2, \ldots, v^{\prime}-1\right\}$, along with $q_{x^{\prime}} r_{v^{\prime}}$ $\forall x^{\prime} \in\{1,2, \ldots, m\}$.

From Proposition \ref{remark-c2=c1+1}, we have that the first $\left(s-1\right) d+l$ elements of $(\mathbf{q} \otimes \mathbf{r})^{\downarrow}$ consist of $q_{x} r_{y}$, $\forall x \in\{1,2, \ldots, d\}$ and $\forall y \in\left\{1,2, \ldots, s-1\right\}$, along with $q_{x^{\prime}} r_{s}$ $\forall x^{\prime} \in\{1,2, \ldots, l\}$ when $q_{l} r_{s} \geqslant q_{1} r_{s+1}$, $q_{d} r_{s-1} \geqslant q_{l+1} r_{s}$ and $q_{d} r_{s-1} \geqslant q_{1} r_{s+1}$, which is true for $\forall l \in \cal{L}$. Hence, Proposition \ref{remark-c2=c1+1} fills such loopholes in the Theorem 1 in \cite{PhysRevA.99.052348}, and also gives finer characterizations of $\mathbf{r}$ by increasing the number of constraints on $\mathbf{r}$.

\subsection{Proof of corollary \ref{coro-dual}}
\begin{proof}
	By the definition of $M_{\cL}$, we have:
	\bea\nonumber
	\begin{aligned}
		\sum_{i=M_{\cL}+1}^{d}p_i\leqslant\sum_{i=M_{\cL}+1}^{d}q_i, \\
		\sum_{i=M_{\cL}+2}^{d}p_i>\sum_{i=M_{\cL}+2}^{d}q_i.
	\end{aligned}
	\eea
	This implies that $p_{M_{\cL}+1}<q_{M_{\cL}+1}$. Combining with $p_d\geqslant q_d$, we get
	\bea\nonumber
	\frac{p_d}{p_{M_{\cL}+1}}>\frac{q_d}{q_{M_{\cL}+1}}.
	\eea
If $\frac{q_d}{q_{M_{\cL}+1}}\geqslant\frac{r_k}{r_{k-1}}$, we have the following partial order the vector $\mathbf { q } \otimes \mathbf { r }$,
	\bea
	\label{A}
	q_dr_k\leqslant q_{d-1}r_k\leqslant\dots\leqslant q_{M_{\cL}+1}r_k\leqslant q_dr_{k-1}.
	\eea
For $\frac{p_d}{p_{M_{\cL}+1}}>\frac{q_d}{q_{M_{\cL}+1}}$, we have the following partial order the vector $\mathbf { p } \otimes \mathbf { r }$,
	\bea
	\label{B}
	p_dr_k\leqslant p_{d-1}r_k\leqslant\dots\leqslant p_{M_{\cL}+1}r_k\leqslant p_dr_{k-1}
	\eea
Combining \eqref{A} and \eqref{B}, we have that the last $d-M_{\cL}$ elements of $( \mathbf { p } \otimes \mathbf { r } ) ^ { \downarrow }$ and $( \mathbf { q } \otimes \mathbf { r } ) ^ { \downarrow }$ are $p_dr_k$, $p_{d-1}r_k$, $\dots$,  $p_{M_{\cL}+1}r_k$ and  $q_dr_k$, $q_{d-1}r_k$, $\dots$, $q_{M_{\cL}+1}r_k$, respectively. Therefore, we have
	\bea\begin{aligned}\nonumber
		&\sum_{i=dk-d+M_{\cL}+1}^{dk}( \mathbf { p } \otimes \mathbf { r } )_i ^ {\downarrow}=r_k\sum_{i=M_{\cL}+1}^{d}p_i\\&\leqslant r_k\sum_{i=M_{\cL}+1}^{d}q_i=\sum_{i=dk-d+M_{\cL}+1}^{dk}( \mathbf { q } \otimes \mathbf { r } )_i ^ {\downarrow},
	\end{aligned}
	\eea
which implies that
	\bea\nonumber
	\sum_{i=1}^{dk-d+M_{\cL}}( \mathbf { p } \otimes \mathbf { r } )_i ^ {\downarrow}> \sum_{i=1}^{dk-d+M_{\cL}}( \mathbf { p } \otimes \mathbf { r } )_i ^ {\downarrow},
	\eea
and contradicts the assumption that $\mathbf { p } \otimes \mathbf { r } \prec \mathbf { q } \otimes \mathbf { r }$.
\end{proof}

\end{document}